\documentclass[letterpaper,11pt]{article}
\usepackage{mathpazo}
\usepackage{nicefrac}
\usepackage{rotating}
\usepackage{setspace}
\usepackage{hyperref}
\usepackage{relsize}
\usepackage{hyperref}
\hypersetup{colorlinks=true,linkcolor=black,citecolor=blue}
\usepackage{color}
\usepackage{amsthm}
\usepackage{amsmath}
\usepackage{enumerate}
\usepackage{amssymb}
\usepackage{bbm}

\usepackage{thmtools}
\usepackage{fullpage}
\usepackage{cleveref}
\usepackage[classfont=bold,
langfont=roman,
funcfont=italic]{complexity}

\newtheorem{theorem}{Theorem}[section]

\newtheorem{lemma}[theorem]{Lemma}

\newtheorem{claim}[theorem]{Claim}
\newtheorem{open}[theorem]{Open Problem}

\newtheorem{definition}[theorem]{Definition}
\newtheorem{remark}[theorem]{Remark}

\usepackage{mleftright}

\newcommand{\zo}{\{0,1\}}
\renewcommand{\S}{\mathcal{A}}
\newcommand\nk{\binom{[n]}{k}}

\newcommand{\Test}{\mathcal{T}}

\renewcommand{\D}{V}

\newcommand{\abs}[1]{\left\lvert #1 \right\rvert}
\renewcommand{\G}{V}

\renewcommand{\DP}{\textsf{DP}}
\newcommand{\DPV}{\textsf{DP}_V}

\newcommand{\maj}{\textsf{maj}}

\newcommand{\sett}[2]{\{#1 | #2\}}
\renewcommand{\d}{\textsf{dec}}

\newcommand{\vecOne}[1]{\mathbf{1}_{#1}}
\newcommand{\ip}[2]{\langle{#1, #2}\rangle}
\newcommand{\norm}[1]{\left\lVert#1\right\rVert}
\title{\textbf{Towards a General Direct Product Testing  Theorem }}
\date{}
\author{ \and Karthik C.\ S.}
\author{Elazar Goldenberg\vspace{0.1cm}\\
 The Academic College of Tel Aviv-Yaffo\vspace{0.1cm} \\
\texttt{elazargo@mta.ac.il} \vspace{0.1cm}
\and 
Karthik C.\ S.\footnote{This work was partially supported by Irit Dinur's  ERC-CoG grant 772839.   } \vspace{0.1cm}\\
 Weizmann Institute of Science\vspace{0.1cm}\\
   \texttt{karthik.srikanta@weizmann.ac.il}
}
\begin{document}
\maketitle  

\begin{abstract}
The Direct Product encoding of a string $a\in \{0,1\}^n$ on an underlying domain $V\subseteq \nk$, is a function $\DP_V(a)$ which gets as input a set $S\in \D $ and outputs $a$ restricted to $S$.  In the Direct Product Testing Problem, we are given a function $F:\D\to \{0,1\}^k$, and our goal is
to test whether $F$ is close to a direct product encoding,
i.e., whether there exists some $a\in \{0,1\}^n$  such that on most sets $S$, we have
$F(S)=\DP_V(a)(S)$.
A natural test  is as follows: select a pair $(S,S')\in \D$ according to some underlying distribution over $V\times V$, query $F$ on this pair, and check  for consistency on their intersection. 
Note that the above distribution may be viewed as  a  weighted graph over the vertex set $V$ and is referred to as a test graph.

The testability of direct products was studied over various domains and test graphs: Dinur and Steurer (CCC'14) analyzed it when $V$ equals the $k$-th slice of the Boolean hypercube and the test graph is a member of the Johnson graph family.  Dinur and Kaufman (FOCS`17) analyzed it for the case where $V$ is the set of faces of a Ramanujan complex, where in this case $V=O_k(n)$. In this paper, we study the testability of direct products in a general setting, addressing the question: what properties of the domain and the test graph allow one to prove a direct product testing theorem?

Towards this goal we introduce the notion of coordinate expansion of a test graph. Roughly speaking a test graph is a coordinate expander if it has global and local expansion, and has certain nice intersection properties on sampling. We show that whenever the test graph has coordinate expansion then it admits a direct product testing theorem.  Additionally, for every $k$ and $n$ we provide a direct product domain $V\subseteq \binom{n}{k}$ of size  $n$, called the Sliding Window domain for which we prove direct product testability.
\end{abstract}
\clearpage

\section{Introduction}\label{sec:intro}
The direct product encoding of a function is a way to aggregate multiple values of the 
input function using a single query.
Justifying the vague intuition that it is much harder to compute multiple values of a 
function rather then a single value of it, the direct product encoding has been successfully 
used in several contexts of hardness amplification. The hardness can either measure the fraction of
inputs on which every reasonable-time algorithm fails to compute the input function, or the fraction 
of unsatisfied assignments of a given CNF-formula or the communication complexity of the function.

In most of the PCP constructions an assignment to the given input is broken into many tiny pieces. Each small piece is encoded individually and then one should be able to test whether these tiny pieces could be stitched together into a global assignment. This testability task is referred to as an agreement test, and instantiations of it include low degree tests such as the plane vs.\ plane \cite{RS97}, the line vs.\ line test \cite{AS97} and the cube vs.\ cube test \cite{BDN17}, and the direct product test used in \cite{DR06}.

More concretely, we associate the direct product encoding of strings of size $n$, with some underlying domain\footnote{For the ease of presentation, we only consider domains which are a subset of $\binom{[n]}{k}$ in this section. However, in the rest of the paper we consider $\D$ which is a collection of subsets of $[n]$, and all our results are proved for this more general case.} 
$\D$ which is a collection of subsets of $[n]$ of cardinality $k$. Given a string $a\in \zo^n$ its direct product encoding on the domain $\D$, denoted by $\DP_V(a)$, is 
defined as follows: For every set $S\in \D$ we define $\DP_{\D}(a)(S)=a|_S$ (where $a|_S$ is the restriction of 
$a$ to the coordinates in $S$). In this paper we study the testability of this encoding, namely: Given 
$F:\D \to \zo^k$ we want to decide whether $F$ agrees with some $\DP_{\D}(a)$ on most sets $S$ while querying $F$ only on a few locations, specifically two. In other words, we focus on two-query tests in the paper where we pick a pair of subsets (both in the domain) according to some fixed distribution and then check if the two subsets agree on their intersection. We say that a domain $V$ admits a direct product testing theorem if there exists a  two-query test $\Test $ satisfying the following: For every $\varepsilon\ge 0$ and $F:\D \to \zo^k$ if $\Test$ accepts $F$ with probability $1-\varepsilon$, then we have $F(S)=
\DP_{\D}(a)(S)$ for some $a\in \zo^n$ on $1-O(\varepsilon)$-fraction of the sets $S$ in $V$, where the constant behind the $O$ notation is independent of $\abs{V}$ and $k$.

This question was studied under various domains. Dinur and Steurer~\cite{DS14} analyzed a two-query test under the domain $\D=\nk$. Recently, Dinur and Kaufman~\cite{DK17} studied this question in a much shrunken domain, which is obtained by considering the set of the faces of a high dimensional expander.  However, both of these proofs are tailored to the structure of their own domain and cannot  be (trivially) generalized  to other domains. It is natural to ask whether a more generalized argument can be applied covering both of these domains, and on which domains it may be applied.   
The main question we are investigating is as follows:

\vspace{3mm}
\begin{center}
\begin{minipage}{.7\textwidth}
\emph{Which domains admit a two-query direct product testing theorem?}
\end{minipage}
\end{center}
\vspace{3mm}

Let us elaborate more about the previous proofs. The proofs given by ~\cite{DS14} and~\cite{DK17} first analyze the testability in the high error regime, i.e. when the acceptance probability is slightly bounded away from $0$. They show that any function that passes the test with non-negligible probability $\varepsilon$ must agree with some legal codeword $\DP_{\D}(a)$ on $\Omega(\varepsilon)$ fraction of sets. Then they analyze the test in the low error regime, i.e. when the acceptance probability of the test is close to $1$. Finally they stitch local  tiny agreements into a single codeword and show that the agreement is almost everywhere. 

We would like to establish a direct product testing theorem using a more straightforward approach: we decode a string from the input function $F$ using the majority operator and then show that if the test passes with high probability then $F$ is close to the direct product encoding of the decoded string. More precisely, given the input function $F$, we define a string $a\in \zo^n$ as follows: for every coordinate $i\in [n]$ we set $a_i$ to be the majority value of $F(S)_i$, where the majority is taken over the sets that contain $i$. Next we show that if $F$ passes the test with probability $1-\varepsilon$ then $F$ must be  $1-O(\varepsilon)$-close to $\DP_V(a)$. We remark that Dinur and Reingold~\cite{DR06} indeed followed this proof strategy, however, their proof admits only a relaxed notion of closeness between the input function and the direct product encoding of the decoded string (namely, that on most sets $S$, $F(S)$ and $\DP_V(a)(S)$ agree \textit{only} on most of the coordinates in $S$).

Observe that any two-query test on a domain $V$ gives rise to a weighted graph whose vertex set is $V$ and the weight we assign for each pair $(S,S')$ is the probability of this pair being picked by the test\footnote{In this paper we analyze test graphs which are undirected.}. We refer to this graph as the test graph. We say that a test graph yields a tester for the domain $V$, if for every $\varepsilon\ge 0$ and every function $F:V\to \zo^k$ the following holds: if the test accepts $F$ with probability $1-\varepsilon$, then $F$ must be $1-O(\varepsilon)$-close to some $\DP_V(a)$. Here the test corresponds to picking an edge $(S,S')$ at random  (according to the distribution of weights on the edges) and accepting if and only if $F(S)|_{S\cap S'}=F(S')|_{S\cap S'}$.	

Another proof insight that we desire is the explicit use of the properties of the underlying test graph. For example, one property that the test graph must satisfy to be a tester is that for most edges $(S,S')$ the intersection between $S$ and $S'$ is linear in $k$. Assume not, then we consider the following construction of $F$: We start from $F=\DP_V(a)$ for some $a\in \zo^n$ and then for each $S\in V$ we reset  the value of $F(S)_i$ for some random $i\in S$. Then for most sets $(S,S')$ with small intersection the test accepts but $F$ is far from any direct product codeword.
Another property that the test graph must have is some notion of expansion. Summing up, our more refined question is as follows:

\vspace{3mm}
\begin{center}
\begin{minipage}{.73\textwidth}
\emph{What properties of the test graph yields a tester for its underlying domain? }
\end{minipage}
\end{center}
\vspace{3mm}

\subsection{Our Results}

\begin{sloppypar}Our conceptual contributions in this paper are two-fold. First, we introduce a notion called \emph{coordinate expansion} which captures the properties of direct product testable domains. Second, we introduce the sliding window domain which is of size exactly \emph{equal} to the universe and is direct product testable. Our main technical contribution is showing that domains having coordinate expansion with certain parameters admit a direct product theorem. \end{sloppypar}

\subsubsection{A General Direct Product Theorem}

We introduce below the notion of coordinate expansion. Informally, a coordinate expander has both global and local expansion properties, and has good intersection properties. 

\begin{definition}[$(\lambda,\rho)$-Coordinate   Expander]\label{def:CE}
	Let $G=(V,E)$ be a test graph, where $V\subseteq \nk$. For $i\in [n]$ let $V_i=\sett{S\in V}{i\in S}$ and $G_i$ be the subgraph of $G$ induced by the vertices in $V_i$. 		The graph $G$ is called $(\lambda,\rho)$-coordinate   expander if:
	\begin{enumerate}
		\item $\lambda(G)<\lambda$ (where $\lambda(G)=\max \{\abs{\lambda_2(A_G)}, \abs{\lambda_{\abs{V}}(A_G)}\}$ and $A_G$ is the normalized adjacency matrix of $G$).
		\item For every $i\in [n]$ we have that $\lambda(G_i)<\lambda$ and for each $S\in V_i$ the probability that a uniformly  random neighbor $S'$ of $S$ is in $V_i$ is at least $\rho$.
		\item For every subset $S$ and $T\subseteq S$, satisfying $\abs{T} \ge 2 /\rho$ , the probability that for a uniformly  random neighbor $S'$ of $S$ we have $\abs{S'\cap T} \le  \rho\abs{T}/2$ is upper bounded bounded by $\lambda$ .
	\end{enumerate} 
\end{definition}

Notice that condition 1 implies that the test graph must be a good expander (in the traditional sense). Moreover, condition 2, implies that on certain local subsets (i.e., subsets containing a common coordinate) of vertices, the induced subgraph must be expanding as well.  Finally, condition 3 implies that the neighbors of every subset $S$ samples well every subset $T$ of $S$. 

Observe that condition 2 is necessary for the test graph to be a direct product tester. To see this, consider a test graph that does not satisfy this property, namely, there exists a coordinate $i\in [n]$ for which: there exits a set $B_i\subset V_i$ such that $\Pr_{S'\in V_i}[S'\notin B_i| S\in B_i]=o(1)$. Then, we show that the test graph does not yield a tester. Indeed, consider the following construction of $F$: we first choose $F=\DPV(a)$ for some $a\in \zo^n$. Then for every $S\in B_i$ we change the value of $F(S)_i$ to  $1-a_i$. Clearly, the distance of $F$ from a direct product encoding equals $\delta:=\abs{B_i}/\abs{V}$. However, the rejection probability equals: 
$$2\cdot \Pr_{S'\sim S}[S\in B_i\text{ and } S'\in V_i \setminus B_i]\le 2\cdot \Pr[S\in B_i]\cdot \Pr[S'\in V_i] \cdot \Pr_{S'\in V_i}[S'\notin B_i| S\in B_i]=o(1)\cdot \delta.$$

Then, we show our main technical result that coordinate expansion implies direct product testing (for a certain range of parameters).

\begin{theorem}\label{thm:main}
	Let $\rho\ge 1/2$ and $\lambda\le 1/33$. Let $G=(V,E)$ be a test graph, $V\subseteq \nk $, let $\varepsilon\ge 0$, and $F:V\to \{0,1\}^k$. 	
	Let $G$ be a $(\lambda,\rho)$-coordinate   expander. If $F$ passes the test implied by the test graph $G$ with probability $1-\varepsilon$ then $F$ is $1-O(\varepsilon)$-close to $\DP_V(a)$ for some $a\in\{0,1\}^n$.
\end{theorem}

The overview of the above proof is given in Section~\ref{sec:tech}. Also, as an application of
the above theorem, we show\footnote{The claim as written here is slightly inaccurate. Please refer to Appendix~\ref{sec:apphalf} for a precise statement.} a direct product theorem for the test graph isomorphic to the Johnson graph $J(n,k)$ when $k$ is close to $n/2$, where $J(n,k)$ is a graph whose vertex set is the set of all subsets of $[n]$ of cardinality $k$, and two subsets have an edge if their intersection is equal to $k/2$. 
This should be compared to \cite{DS14}, where they show the direct product for the Johnson graph for all the layers up to $n/2$ (i.e., for all $J(n,k)$ where $k\le n/2$).

The main open problem stemming from our work is to improve the parameters in Theorem~\ref{thm:main}. In particular, does the following hold?
\begin{open}\label{open1}
Does $(\nicefrac{1}{2},\nicefrac{1}{2})$-coordinate expansion imply a direct product theorem?
\end{open}

A positive resolution of the above open question would imply direct product testability on the test graph isomorphic to the Johnson graph for every layer of the Boolean hypercube (completely recovering the results in~\cite{DS14}). It even implies a direct product testability on a new domain: where the subsets are stemming from $d$-dimensional subspaces of $\mathbb{F}_2^m$ and two subsets are connected by an edge if they intersect on a $(d-1)$-dimensional subspace (this is referred to as the Grassmann graph). Finally, we would like to  recall that Theorem~\ref{thm:main} states that $(\nicefrac{1}{33},\nicefrac{1}{2})$-coordinate expansion implies a direct product theorem, i.e., in order to positively resolve Open~Problem~\ref{open1}, we might need to improve the analysis in the proof of Theorem~\ref{thm:main} to accommodate test graphs with weaker expansion properties. 

In fact, if we can resolve Open~Problem~\ref{open1} in a slightly stronger way i.e., if we show that for some small enough constant $\gamma>0$, we have $(\nicefrac{1}{2}+\gamma,\nicefrac{1}{2})$-coordinate expansion implies a direct product theorem then we recover the testability result of~\cite{DK17} on Ramanujan complexes. Summarizing, we view the study of coordinate expansion as providing a unified framework to prove direct product theorems.  Also, it might be useful in the future to establish direct product testability for new domains (in a black-box manner).

\subsubsection{Sliding Window Domain }\label{sec:intro-sliding}

In this subsubsection, we define a new direct product testable domain which we call the sliding window domain, and also discuss about the necessary and sufficient structure that a domain (and test graph) should have, in order to admit direct product testing.

For every $n,k$, the sliding window domain $\S\subseteq\nk$ is the collection of all contiguous $k$-sized subsets (windows) of $[n]$, i.e., $\S=\{\{i,\ldots, i+k-1\}\mid i\in [n]\}$, where the addition is done modulo $n$. Two vertices (i.e., subsets in $\S$) have an edge in the test graph, if their intersection is non-empty. Notice that $|\S|=n$ and yet we show that it admits a direct product theorem (see Theorem~\ref{thm:sliding} for a simple proof).

Let us put the above result in context with the recent breakthrough of Dinur and Kaufman \cite{DK17}. 
In \cite{DK17}, the authors obtain a direct product testable domain (subset of $\nk$) of size $O({2^{k^2}n})$. The domain arises from the highly non-trivial object called Ramanujan complex. Such a domain is studied because apart from admitting a direct product theorem over a domain of size linear in the universe (i.e., $n$), it also has other desirable properties such as \emph{distance amplification} which are needed for applications in gap and hardness amplification. Thus, our direct product testing result (Theorem~\ref{thm:sliding}) provides a conceptual clarification that if one is only interested in direct product testing as a property testing question, then there is a very simple domain of size $n$, namely the sliding window domain, which is testable. 

Roughly speaking, a domain (subset of $\nk$) has distance amplification if for every two strings of relative distance $\delta$, the relative distance between their direct product encoding is $\Omega(k\delta)$. This seems to be a crucial property for PCP applications of direct product testing. 
Thus, the construction of the sliding window domain provides a conceptual clarification as to why we need high dimensional expanders: we can obtain direct product testing from simple constructions like the sliding window domain and we can obtain distance amplification from known constructions of vertex expanders (see Appendix~\ref{sec:distamp}); but to obtain both simultaneously, \cite{DK17} needed high dimensional expanders.
We leave it as an open question whether there exists a simple construction admitting both direct product testability and distance amplification.

\begin{open}\label{open3}
Is there a (relatively) simple domain of linear size in the universe (i.e., $O_k(n)$) for which we have both direct product testing and distance amplification?
\end{open}

\noindent\textbf{Lack of Global Expansion.}
We would like to now briefly discuss about the minimal structure of the domain (and the test graph) sufficient to prove a direct product theorem. This is highlighted by the sliding window domain, an in particular by the proof of its testability (Lemma~\ref{lem:clique} to be precise).
Notice that $G_\S$ has very bad edge-expansion/vertex-expansion but is a very good local expander, i.e., the induced subgraph containing any particular coordinate has good expansion (in fact is a clique). Lemma~\ref{lem:clique} guarantees that in such situations\footnote{Lemma~\ref{lem:clique} can be generalized to accommodate test graphs which are locally subgraphs that strongly satisfy the expander mixing lemma.} the domain admits direct product testing if for every vertex in the test graph, and every element in that vertex, the probability of retaining that element when moving to a uniformly random neighbor is bounded from below by a positive constant. The probability of retaining a coordinate when moving to a random neighbor is $\nicefrac{1}{2}$ in $\S$, and thus $\S$ admits a direct product theorem. Therefore, $\S$ demonstrates that direct product testing does not require the test graph to be an expander (like the Johnson/Ramanujan graph) but only needs to have certain local expansion properties.
Finally, recall that we had earlier argued that local expansion is necessary (to justify the need for condition~2 in Definition~\ref{def:CE}) for direct product testing. 

Finally, it seems that conditions 1 and 3 in coordinate expansion are not (necessarily) needed for direct product testing, but are merely artifacts of our proof (Theorem~\ref{thm:main}). However, these conditions might imply distance amplification\footnote{This would be an interesting question to resolve in either direction.} and are typically guaranteed in structured domains of interest (namely, Johnson, Grassmannian, and Ramanujan).

\subsection{Technical Contribution: Proof Overview of Theorem~\ref{thm:main}}\label{sec:tech}
For the sake of convenience, through out this subsection, we fix $V=\nk$ and the test would pick pairs $(S,S')$ that intersect on $k/2$ elements and checks for agreement.
As suggested above there is a natural way to decode any function $F:V\to \zo^k$ using the majority operator: define a string $a\in \zo^n$ by setting $a_i$
to be the majority value of $F(S)_i$ for all $S\ni i$.
We define $B=\{S|F(S)\neq \DP_V(a)(S)\}$, i.e., $B$ is the subset of the domain that disagrees with the direct product encoding of the decoded string. Also for $S\in B$ we call $i\in S$ conflicting if $F(S)_i\neq a_i$.  Our goal is to show that the test rejects with probability $\Omega(\abs B/\abs{V})$ as $|B|/|V|$ is the relative distance between $F$ and $\DP_V(a)$.

Indeed fix $S\in B$, then it must contain at least one conflicting coordinate, say $i$. Observe that with probability $1/2$ we also have that $i\in S'$. Now if $S'$ were a random set containing $i$, then since at least half
of the elements that contain $i$ agree with the majority value, the test rejects with probability $1/2$. And the overall rejection probability of the test would be at least $\frac{\abs B}{4\abs{V}}$ and we are done. 

However, $S'$ is not a random set that contains $i$, it intersects with $S$ on further $k/2-1$ coordinates. Therefore, it may well be that among the neighbors of $S$ that contain $i$ we do not see the majority value so often. A natural way to overcome this is by aggregating all $S$s' that contain $i$ and disagree with the majority value on $i$. We could try to show that if we start from some member of this set then with constant probability we reach $S'$ that contains $i$ and resides outside of this aggregated set (by using the local expansion property). But this leads into another problem: using this argument sets $S$ that contain many conflicting coordinates are counted many times, whereas sets that contain few conflicting coordinates are counted much less. 

Our analysis proceeds by studying the variance of the number of conflicting coordinates in the following manner. We first sort the set $B$ based on the number of their conflicting coordinates. Let $B_L$ (resp.\ $B_H$) be the first (resp.\ last) third of the elements in $B$ according the sorting. We first show that if the number of conflicting coordinates of each member in $B_L$ is much smaller than it is in $B_H$, then the test rejects with probability $\Omega(\frac{\abs B}{\abs{V}})$. To show this, we prove that whenever the test picks $S\in B_H$ then with constant probability $S'$  is in $B_L \cup \{V\setminus B\}$ (by using the global expansion property).
Moreover, there is a large subset $\Gamma$ of conflicting coordinates in $S$ which are also in $S'$ (follows from condition~3 in Definition~\ref{def:CE}). However, $S'$ has few conflicting coordinates in total (by our choice of $S'$), and thus, there must be a coordinate in $\Gamma$ that agrees with the majority value on $S'$ but disagrees on it on $S$ and hence the test rejects the edge $(S,S')$.

 On the other hand, if the number of conflicting coordinates does not vary a lot among these sets, then we analyze the test by selecting (at random) a single conflicting coordinate in $S$ and focusing on the rejection probability based only on the value of the selected coordinate. 

\subsection{Related Work}
The question of testing the direct product was studied extensively when the underlying domain $\D=\nk$~\cite{Gs97,DR06,DG08,DS14,IKW12}. 
In this setting,  Goldreich and Safra~\cite{Gs97} proposed a constant query test. Dinur and Reingold \cite{DR06} suggested the two-query test mentioned above and analyzed it in the high acceptance regime but with a relaxed  distance measure.

The state of the art in this context is 
the result of Dinur and Steurer\footnote{The result in \cite{DS14} is stated in the language of tuples, i.e., the domain is a subset of $[n]^k$, but their result also holds when the domain is a collection of $k$-sized subsets of $[n]$. See \cite{DDGKS17} for more details.} \cite{DS14} dealing with the domain $\D=\nk$ where $k$ varies between $2$ and $n/2$. They analyze the aforementioned two-query test with $k/2$-intersection size. They analyze it in the high acceptance regime and show that $\nk$ indeed admits a direct product testing theorem. The proof is quite involved and in particular analyzes first the low acceptance regime. Recently, in a breakthrough paper, 
Dinur and Kaufman~\cite{DK17} analyzed the two-query test when the underlying domain is obtained from the set of faces of a Ramanujan complex.
Their approach crucially relies on the result of~\cite{DS14}. More recently Dinur et al.\ \cite{DHKLT19} introduced the notion of \emph{double samplers} and remarked that it might admit a direct product theorem.

We remark that the direct product testability question was further analyzed in the low acceptance regime under the domain $\nk$, see~\cite{DG08,IKW12,DN17} and also under the domain where the universe is $\mathbb{F}_2^m$, and the domain is the set of all subspaces of $\mathbb{F}_2^m$ \cite{IKW12}.

\subsection{Organization of the Paper}
Section~\ref{sec:prel} lists the notations and technical tools that we use in the paper. In Section~\ref{sec:setting} we formalize the notion of direct products and their testing. In Section~\ref{sec:main} we prove our main technical result, namely, that whenever the underlying test graph is a $(\lambda,\rho)$-coordinate  expander it admits a direct product testing theorem. 
Finally, in Section~\ref{sec:slidingWindow} we introduce the sliding window domain for which we show a direct product theorem.  

\section{Preliminaries}\label{sec:prel}

In this section, we list the notations and technical tools used in this paper.

\noindent\textbf{Notations.} We use the following notations throughout the paper. We denote the set $\{1,\ldots ,n\}$ by $[n]$. For any $n,k\in\mathbb{N}$, with $k\le n$, we denote by $\binom{[n]}{k}$, all subsets of $[n]$ of cardinality $k$.
For any set $S$, we denote by $\mathcal{P}(S)$ the power set of $S$, i.e., the set of all subsets of $S$. 
For any graph $G(V,E)$ and any two subsets $S,T\subseteq V$, we denote by $E(S,T)$ the set of all edges between $S$ and $T$. For any $x,y\in\{0,1\}^n$, we denote by $\Delta(x,y)$ the relative Hamming distance between $x$ and $y$ given by the fraction of coordinates in which $x$ and $y$ differ.

\noindent\textbf{Johnson Graph Family.} For every $n,k,t\in\mathbb{N}$ such that $t\le k\le n$, $J(n,k,t)$ is a graph which is a member of the Johnson graph family, whose vertex set is  $\binom{[n]}{k}$, and whose edge set is $\{(S,S')\mid S,S'\in \binom{[n]}{k}, |S\cap S'|= t\}$.

\noindent\textbf{Expander Mixing Lemma.} 
The following is a standard claim concerning the expansion of two sets in expander graphs. For completeness we include a proof in Section~\ref{sec:appendix}:
\begin{claim}
	\label{claim:expanderMixing}
	Let $G=(V,E)$ be a $d$-regular graph and $A$ be its adjacency matrix.  Let $\lambda$ be its second largest eigenvalue in absolute value.
	Let $S,T\subseteq V$ satisfying: $\abs {S}\le \abs {V}/2$ then:
	\[ \Pr_{(u,v)}[v\in T|u\in S ]\le \frac{\abs{T}}{\abs {V}}+\frac \lambda d\sqrt{\frac{ \abs{T}}{\abs{S}}},  \]
	where the probability is  given by first picking $u$ uniformly at random from $S$, and then picking $v$ according to $A$. 	
	Furthermore, let $\mu$ be a distribution on $S$ satisfying that for every two elements $b,b'\in S$:
	$ \mu(b)\le c\mu(b')$, then:
\[\Pr_{(u,v)} [v\in T\mid u\sim\mu] \le 
\frac{|T|}{|V|} +\frac \lambda d \cdot \sqrt{\frac{c \abs{T} }{\abs{S}}} \]
\end{claim}

\section{Direct Product Testing: The Setting}\label{sec:setting}
In this section, we formalize the notion of direct products and their testing. Specifically, we formalize the notion of direct product testing through test graphs, which is slightly non-standard but it helps  in  introducing the notion of coordinate expansion in a later section  succinctly.

For every subset $S$ of $[n]$, let $\mathcal{F}_S$ be the class of all functions whose domain is $S$ and range is $\{0,1\}$.
Let $\G\subseteq \mathcal{P}([n])$ be the domain of the direct product. Let $\mathcal{F}_V$ be the class of all functions whose domain is $\G$ and maps every subset $S$ in $\G$ to a function in $\mathcal{F}_S$. The direct product encoding is a function $\DPV:\{0,1\}^n\to \mathcal{F}_V$ defined as follows: for every string $a\in\{0,1\}^n$, and every subset $S\in \G$, let $\DP_{V}(a)_S$ be defined as the projection function which maps $S$ to $a_S$, the string $a$ restricted to only the coordinates in $S$.

\begin{definition}
	For two functions $F,G\in \mathcal{F}_V$ we define their relative distance as: $$\Delta(F,G)=\frac{\abs{\{S\in\G|F(S)\neq G(S)\}}}{|\G|}. $$ 
	For a function $F$ and a set of functions $\tilde{G}$ we define the distance between $F$ and $\tilde{G}$ as the minimal distance between $F$ and some function $G\in \tilde{G}$. If $\Delta(F,\tilde{G})\le \delta$, we say that $F$ is $1-\delta$-close to $\tilde{G}$, otherwise,
	it is $\delta$-far from $\tilde{G}$.
\end{definition}

For every function $F\in\mathcal{F}_V$, we define $\d(F)$ as follows: Given $F$ construct $a^F\in \zo^n$ in the following way,
	$$a_i^F= \underset{\substack{S\in\G \\ S\ni i}}{\maj}(F(S)_i).$$
	Then, we define $\d(F):=\DP_{\G}(a^F)$.

Let $G_{\G}$ be a graph whose vertex set is $V$. Then we interpret $G_{\G}$ as a test graph on functions defined on $\mathcal{F}_V$ in the following sense: 

\begin{center}
\fbox{%
	\parbox{0.62\textwidth}{%
		\noindent\textbf{Test $\mathcal T(G_{\G})$}:
		\par
		\noindent \textbf{Input}: A function $F\in\mathcal{F}_V$.
		\par
		\noindent \textbf{Procedure}: 
		 Pick an edge $(S,S')$ in $G_V$ uniformly at random.
		
				\noindent \textbf{Output}: Accept if and only if $F(S)|_{S\cap S'}=F(S')|_{S\cap S'}$.

}}
\end{center}

It is important to note that we allow self loops and multiple edges between a pair of vertices. Also, we can generalize the above direct product testing setting to the case when $V$ is a multiset of $\mathcal{P}([n])$, and the results in this paper still hold. However, we choose not to handle this more general setting for the sake of clarity of presentation. The above remark also applies to the case of studying test graphs which are not regular in degree, that are not considered in this paper. 
Finally, throughout the paper, we drop the subscript $V$ in $G_V$, if $V$ is clear from the context.

\section{Direct Product Testing: Coordinate Expansion}
\label{sec:main}
In this section we prove our main technical result, namely, that whenever the underlying test graph is a $(\lambda,\rho)$-Coordinate  Expander (defined next) it admits a direct product testing theorem. 
\begin{definition}[$(\lambda,\rho)$-Coordinate   Expander]
	Let $n\in \mathbb{N}$ and let $G=(V,E)$ be a test graph, where $V\subseteq \mathcal{P}([n])$. For $i\in [n]$ let $V_i=\sett{S\in V}{i\in S}$ and $G_i$ be the subgraph of $G$ induced by the vertices in $V_i$. 	
Let $\lambda(G)=\max\{\abs{\lambda_2(A_G)},\abs{\lambda_{\abs{V}}(A_G)}\}$, where $A_G$ is the normalized adjacency matrix of $G$.
	The graph $G$ is called $(\lambda,\rho)$-coordinate   expander if:
	\begin{enumerate}
		\item $\lambda(G)<\lambda$ and for every $i\in [n]$ we have $\lambda(G_i)<\lambda$.
			\item For every $i\in [n]$ and for each $S\in V_i$ we have $\underset{S'\sim S}{\Pr}[S'\in V_i]\ge\rho$.
		\item For every subset $S$ and $T\subseteq S$, satisfying $\abs{T} \ge 2 /\rho$ , we have
		$\underset{S'\sim S}{\Pr}[|S'\cap T|\leq \rho |T|/2]\leq \lambda$. 	\end{enumerate} 
	\end{definition}

Informally, a domain is a coordinate expander if the test graph is an expander and every induced subgraph of the test graph containing a fixed coordinate is also an expander\footnote{Actually, the property of an expander that we need is that for any two sets of vertices $S,T$ in the graph, the number of edges between $S$ and $T$ is roughly equal to $\alpha|S||T|$, where $\alpha$ is the density of the edge set of the graph.}, and it has good correlation/intersection properties -- i.e., for any subset $S$ and coordinate $i\in S$, an uniformly random neighbor of $S$ contains $i$ with constant probability (say $\rho>0$), and for every $S$ in the domain, and any subset $T$ of $S$, the number of elements of $T$ that we see in a random neighbor of $S$ is close to the expected number, which is $\rho\cdot |T|$.
Below, we see that coordinate expansion of the test graph implies a direct product theorem for the underlying domain.

\begin{theorem}\label{thm:expanders}
	Let $n\in \mathbb{N}$,  and let $\rho\ge 1/2$ and $\lambda\le 1/33$ be some constants. Let $G=(V,E)$ be a graph, $V\subseteq \mathcal{P}([n])$, let $\varepsilon\ge 0$, and $F\in\mathcal{F}^V$. 	
	Let $G$ be a $(\lambda,\rho)$-coordinate   expander. If $F$ passes $\Test(G)$ with probability $1-\varepsilon$ then $F$ is $1-O(\varepsilon)$-close to $\d(F)$.
\end{theorem}

\begin{proof}	Let $F^*:=\d(F)=\DPV(a^F)$. We define $B,C\subseteq V$  as follows:
	$$ B=\{S\mid F(S)\neq F^*(S)\} \text{ and } C=V\setminus B.$$ 
	Let $\beta= \abs{B}/\abs{V}$. 
	Given a subset $S\in V$ we say that a coordinate $i$ is conflicting if the value of $F(S)$ at $i$ does not equal $a_i^F$. For a set $S$ denote by $B(S)$ the set of conflicting coordinates in $S$. We show that $\Test(G)$ rejects with  probability at least $\Omega(\beta)$.
		
Let us sort in ascending order the elements of $B$ based on the number of coordinates on which they disagree with $F^*$. 
For a parameter $0\le p \le 1$ we define the set $B_{\ge p}$  as the set of last $(1-p)\abs{B}$ elements of $B$ (and similarly  the set $B_{\le p}$ is the set of the first $p\abs{B}$ elements of $B$).
We denote by $m_p$ the number of conflicting coordinates of the $p\abs{B}$-th element of $B$.

Let $c=3/40$. We consider two cases based on  $m_{c}, m_{1/2}$ and $m_{1-c}$.

\paragraph{Case 1: $m_{1-c}>\frac{2}{\rho}m_{1/2}$ or $m_{1/2}>\frac{2}{\rho}m_{c} $:}

For both the possibilities we have similar arguments, which is why they are clubbed under one case, but will be handled separately for ease of presentation.

\paragraph{Case 1A: $m_{1-c}>\frac{2}{\rho}m_{1/2}$:}

The probability that an uniformly random $S\in V$ is in  $B_{\ge 1-c}$ is $c\beta$. Now by   Claim~\ref{claim:expanderMixing}, we get that
$$ \Pr [S'\in B_{> 1/2} | S\in B_{\ge 1-c}] < \beta/2 +\lambda\sqrt{\frac{ 1}{2c}}, $$
so with probability at least $1-\beta/2 -\lambda\sqrt{\frac{ 1}{2c}}$ if $S\in B_{\ge 1-c}$ then $S'\in B_{\le 1/2} \cup C$. 

Now, by the third property of  $( \lambda,\rho)$-coordinate   expander, the probability that $\abs{S' \cap B(S)}\le  \frac{\rho}{2} |B(S)|$ is at most $\lambda$. Notice that the probability that $\abs{S' \cap B(S)}\le  m_{1/2}$ is at least the probability that $\abs{S' \cap B(S)}\le  \frac{\rho}{2} |B(S)|$ (because $m_{1/2}<\frac{\rho}{2} m_{1-c}\le \frac{\rho}{2} |B(S)|$). Hence we have that the probability that $\abs{S' \cap B(S)}\le m_{1/2}$ is at most $\lambda$.

Overall, using union bound,  conditioned on $S\in B_{\ge 1-c}$, the probability that $S'\in B_{\le 1/2} \cup C$ and $\abs{S' \cap B(S)}>m_{1/2}$ is at least $1-\beta/2 -\lambda\sqrt{\frac{ 1}{2c}}-\lambda$. But in such a case since $S'\in B_{\le 1/2} \cup C$ we get $|B(S')|\le m_{1/2}$, so there exists at least one coordinate $i\in S \cap S'$ on which $F(S')_i=a_i^F$ but  $F(S)_i\neq a_i^F$, so the test rejects. In total $\Test$ rejects with probability at least $c\beta \left(1-\beta/2 -\lambda\sqrt{\frac{ 1}{2c}}-\lambda\right)\ge c\beta \left(1/2 -\lambda\left(\sqrt{\frac{ 1}{2c}}+1\right)\right)$ (where we used a trivial bound that $\beta\le 1$). Notice that $1/2 -\lambda\left(\sqrt{\frac{ 1}{2c}}+1\right)>0 $ holds for $c=3/40$ whenever $\lambda\le 0.13$. 

\paragraph{Case 1B: $m_{1/2}\ge \frac{2}{\rho}m_{c} $:}
In this case we would like to mimic the proof strategy of the previous case. That is we would like to show that with non-zero constant probability a random neighbor in $B_{\ge 1/2}$ is in $B_{\le c}\cup C$. By an application of Claim~\ref{claim:expanderMixing}, we get:

$$ \Pr [S'\in B_{> c} | S\in B_{\ge 1/2}] <  (1-c)\beta + \lambda\sqrt{2-2c}, $$
so with probability at least $1-(1-c)\beta - \lambda\sqrt{2-2c}$ if $S\in B_{\ge 1/2}$ then $S'\in B_{\le c} \cup C$. 

Now, by the third property of  $( \lambda,\rho)$-coordinate   expander, the $\Pr[\abs{S' \cap B(S)}\le  \frac{\rho}{2} |B(S)|]$ is at most $\lambda$. Notice that $m_c \le \frac{\rho}{2}m_{1/2}\le \frac{\rho}{2}\cdot |B(S)| $ and thus $\Pr[\abs{S' \cap B(S)}\le  \frac{\rho}{2} |B(S)|]\ge \Pr[\abs{S' \cap B(S)}\le m_c]$. Therefore we have $\Pr[\abs{S' \cap B(S)}\le m_c]\le \lambda$.

Overall, using union bound,  conditioned on $S\in B_{\ge 1/2}$, the probability that $S'\in B_{\le c} \cup C$ and $\abs{S' \cap B(S)}>m_c$ is at least $1-(1-c)\beta - \lambda\sqrt{2-2c}-\lambda$. But in such a case since $S'\in B_{\le c} \cup C$ we get $|B(S')|\le m_c$, so there exists at least one coordinate $i\in S \cap S'$ on which $F(S')_i=a_i^F$ but  $F(S)_i\neq a_i^F$, so the test rejects. In total $\Test$ rejects with probability at least $\frac{\beta}{2} \left(1-(1-c)\beta - \lambda\sqrt{2-2c}-\lambda\right)\ge \frac{\beta}{2} \left(c - \lambda\left(\sqrt{2-2c}+1\right)\right)$ (where we used a trivial bound that $\beta\le 1$). Notice that $\left(c - \lambda\left(\sqrt{2-2c}+1\right)\right)>0$ holds for $c=3/40$ whenever $\lambda\le 0.03177$. 

\paragraph{Case 2: $m_{1-c}\le \frac{4}{\rho^2}m_c$:}
Define $B_{(c,1-c)}$ as the set $B\setminus (B_{\le c} \cup B_{\ge 1-c})$.
Observe that in $B_{(c,1-c)}$ the number of conflicting coordinates is between $m_c$ and $4m_c/\rho^2$.
Now we would like to consider a different test $\Test'(G)$ that selects $S,S'$ according to $G$.
If $S\notin B_{(c,1-c)}$ then  $\Test'$ accepts. Otherwise, it picks uniformly at random $i_0\in B(S)$ 
and checks for consistency \emph{only} on $i_0$, namely: 
It rejects iff $i_0 \in S'$ and  
$F(S)_{i_0}\neq  F(S')_{i_0}$. Clearly the rejection probability of $\Test'(G)$ is at most the rejection probability of $\Test(G)$. We conclude the proof by showing that $\Test'(G)$ rejects $F$ with probability $\Omega(\beta)$.

With probability  $(1-2c)\beta$ the test $\Test'$ selects $S\in B_{(c,1-c)}$ and we would like to analyze the rejection probability conditioned on that.
For this sake we bound the probability of the following events:
\begin{itemize}
	\item  $E_1$ is the event where $S'\in B_{\le c} \cup B_{\ge 1-c}$.
	\item $E_2$ is the event where $i_0 \in S'$ and $S' \notin \tilde{B}_{i_0}$ where $\tilde{B}_{i}=\sett{S\in B_{(c,1-c)}}{F(S)_i\neq a_i^F}$.
\end{itemize}

If the event $E_2$ occurs but $E_1$ does not, then it must be the case that $F(S')_{i_0}=a_{i_0}^F$. Hence $\Test'$ rejects. 
As a consequence $\Pr[\Test' \text{ rejects} ]\ge (1-2c)\beta(\Pr[E_2| S\in B_{(c,1-c)}]- Pr[E_1|S\in B_{(c,1-c)}])$. Thus it suffices to show that $(\Pr[E_2| S\in B_{(c,1-c)}]- Pr[E_1|S\in B_{(c,1-c)}])$ is a positive constant bounded away from 0.

To bound the probability for the event $E_1$ we use Claim~\ref{claim:expanderMixing}: The probability of $E_1$ conditioned on $S\in B_{(c,1-c)}$ is at most $ 2c \beta +\lambda \sqrt{\frac{2c}{1-2c}}$.

Since the graph $G$ is a $(\lambda,\rho)$-coordinate expander then for each $i\in S$, we have that $\Pr[i\in S']\ge \rho$, in particular this is true for $i_0$, hence: $\Pr[i_0\in S']\ge  \rho$. 

Now we divide the event $E_2$ into disjoint events depending on the value of $i_0$ and bound the rejection probability of $\Test'$ conditioned on specific value of $i_0$.
Fix $i\in [n]$ and assume that $\Test'$ selects $S,S'\in V_i$ and sets $i_0=i$ (so $S\in \tilde{B}_{i}$). We denote by $\beta_i$ the fraction $\frac{\abs{\tilde{B}_{i}}}{\abs{V_i}}$. Observe that $\beta_i\le 1/2$, since otherwise the majority value would become the value of $F(S)_i$, but we have $S\in \tilde{B}_{i}$.

Note, that under the assumption that $\Test'$ selects $i_0=i$ and $S\in \tilde{B}_{i}$, 
sets $S$ with few conflicting coordinates are more likely to be chosen than those who have many of them. 
However, since by our assumption the number of conflicting coordinates is between $m^*$ and $\frac{4}{\rho^2}m^*$, then sets with $m^*$ conflicting coordinates are only $4/\rho^2$-times more probable than those having $\frac{4}{\rho^2}m^*$-conflicting coordinates. 
Denote by $\mu$ the distribution of picking $S\in \tilde{B}_i$ assuming that $\Test'$ selects $i_0=i$. By an application of Claim~\ref{claim:expanderMixing} we get: 
$$\Pr_{S\sim \mu, S'}[S'\in \tilde{B}_i]\le \frac{\abs{\tilde{B}_i}}{\abs{V_i}}+\lambda\sqrt{\frac{4}{\rho^2 }} \le \frac{1}{2}+2\lambda/\rho$$

So we get that, 
$$\Pr[E_2| S\in B_{(c,1-c)}] = \left(1-\Pr_{S\sim \mu, S'}[S'\in \tilde{B}_i\mid i_0\in S']\right)\cdot \Pr[i_0\in S']\ge  \frac{\rho}{2}-2\lambda.$$

Summing up, we get: 
\begin{eqnarray*}
	\Pr[\Test' \text{ rejects}| S\in B_{(c,1-c)}] &\ge& \Pr[E_2| S\in B_{(c,1-c)}]-\Pr[E_1| S\in B_{(c,1-c)}] \\
	& \ge & \frac{\rho}{2}- 2\lambda-\left(2c +\lambda \sqrt{\frac{2c}{1-2c}}\right)\\
		& \ge & \frac{1}{4}- 2c -\lambda \left(2+\sqrt{\frac{2c}{1-2c}}\right),
\end{eqnarray*}
a constant bounded away from 0 for $c=3/40$ whenever $\lambda<0.04$.
\end{proof}

\begin{remark}\label{rem:1}
Notice that the above theorem holds for any $\lambda,\rho,c\in [0,1]$ satisfying the condition that the below three expressions are a constant bounded away from 0:
$$
1/2 -\lambda\left(\sqrt{\frac{ 1}{2c}}+1\right),\ c - \lambda\left(\sqrt{2-2c}+1\right),\ \frac{\rho}{2}- 2\lambda-\left(2c +\lambda \sqrt{\frac{2c}{1-2c}}\right).
$$
\end{remark}

In Appendix~\ref{sec:apphalf}, we consider the test graph $J(n,k,t)$ and using the above remark show a direct product theorem when $k$ is close to $n/2$ and $t$ is close to $k/2$.

\section{Sliding Window Domain}
\label{sec:slidingWindow}
In this section, we introduce the sliding window domain for which we show a direct product theorem.

\noindent\textbf{Construction.} 
Let $k,n\in\mathbb N$ such that $k\le n$.
Let $\S$ be a collection of $n$ subsets of $[n]$ of Hamming weight $k$. 
$$
\S=\{\{i,\ldots ,i+k-1\}\mid i\in[n]\},
$$
where the addition is done\footnote{Strictly speaking, the addition is done modulo $n$ and then  the resulting number is incremented by one.} modulo $n$. 

\noindent\textbf{Testability.} The domain of our direct product test is $\mathcal{A}$. The corresponding test is as follows:

\begin{center}
\fbox{%
	\parbox{0.6\textwidth}{%
		\noindent\textbf{Test $\mathcal{T}$}:
		\par
		\noindent \textbf{Input}: A function $F:\S \to\{0,1\}^k$.
		\par
		\noindent \textbf{Procedure}: 
		 Pick uniformly at random $S\in\S$. Then pick uniformly at random $S'\in \S$ such that $S\cap S'\neq \emptyset$.
		\par		
		\noindent \textbf{Output}: Accept if and only if $F(S)|_{S\cap S'}=F(S')|_{S\cap S'}$.

}}\end{center}

The test graph $G_{\S}$ of the above is given by the vertex set $\S$ and the edge set $\{(S,S')\mid S\cap S'\neq\emptyset\}$.
The correctness of the above test is shown below. We would like to emphasize that $|\S|=n$ and yet admits a direct product theorem.

\begin{theorem}\label{thm:sliding}
Let $\varepsilon\ge 0$ and $F\in\mathcal{F}_{\S}$. 
If $F$ passes $\Test(G_{\S})$ with probability $1-\varepsilon$ then $F$ is $(1-4\varepsilon)$-close to $\d(F)$.
\end{theorem}
\begin{proof}

We will in fact prove a more general direct product testing result.

\begin{lemma}\label{lem:clique}
Let $n\in \mathbb{N}$ and $G=(V,E)$ be a $d$-regular graph where $V\subseteq \mathcal{P}([n])$, let $\varepsilon\ge 0$, and $F\in\mathcal{F}_V$. 
	For every $i\in[n]$, let the induced subgraph of $V_i$ in $G$ be a clique (with self loops). Additionally, let $c>0$ be a constant such that for every $S\in V$ and every $i\in S$, the probability that a uniformly random neighbor $S'$ of $S$ in $G$ contains $i$ is at least $c$. If $F$ passes $\Test(G)$ with probability $1-\varepsilon$ then $F$ is $(1-\frac{2\varepsilon}{c})$-close to $\d(F)$.
\end{lemma}

Now we show that the above lemma gives the proof of the theorem. Let $\S_i=\{S\in \S\mid i\in S\}$.
Note that for every $i\in[n]$, the induced subgraph of $\S_i$ in $G$ is a clique (with self loops) because any two subsets in $\S_i$ have $i$ in their intersection and thus have non-empty intersection. Also for every $S\in \S$ and every $i\in S$, the probability that a uniformly random neighbor $S'$ of $S$ in $G$ contains $i$ is at least $\nicefrac{1}{2}$. Thus, from Lemma~\ref{lem:clique} the theorem follows.
\end{proof}

We complete the proof of the above theorem by showing Lemma~\ref{lem:clique} below.

\begin{proof}[Proof of Lemma~\ref{lem:clique}]
Let $F^*:=\d(F)=\DPV(a^F)$.
	Let $B\subseteq V$ be defined as follows:
	$$ B=\{S\mid F(S)\neq F^*(S)\}.$$
 Let $C_i\subseteq V_i$ be defined as follows:
	$$ C_i=\{S\in V_i\mid F(S)_i= a_i^F\}.$$ 
	By definition of $a_i^F$, it is clear that $|C_i|\ge |V_i|/2$.

	Since $F$ passes $\Test(G)$ with probability $1-\varepsilon$ this implies that the number of edges that fail $\Test(G)$ is at most $ \varepsilon\cdot \frac{|V| d}{2}$.
	
Fix $S\in B$. Fix $i\in [n]$ (arbitrarily) such that $F(S)_i\neq F^*(S)_i$. Now observe that whenever $S'\in C_i$, the test $T(G)$ rejects the edge $(S,S')$ in $G$ because $F(S)_i\neq a_i^F=F(S')_i$. This implies that there are at least $|C_i|\ge |V_i|/2\ge \nicefrac{cd}{2} $ many edges incident on $S$  that fail the test $T(G)$. Therefore, there are in total at least $|B|\cdot \nicefrac{cd}{4}$ edges that fail the test.
 Recall that the total number of rejected edges is at most $ \varepsilon\cdot \frac{|V| d}{2}$. Thus we have that $|B|/|V|\le \frac{2\varepsilon}{c}$. The proof is concluded by noting that the distance between $F$ and $F^*$ is exactly $|B|/|V|$.	
\end{proof}

Note that Lemma~\ref{lem:clique} holds even when the induced subgraph of $V_i$ in $G$ is a clique without self loops. In Appendix~\ref{sec:simple}, we provide a couple of direct product theorems on domains that are known in literature as an immediate  consequence of this lemma. 

\noindent\textbf{Lack of Global Expansion.}
Notice that $G_\S$ has very bad edge-expansion/vertex-expansion but is a very good local expander, i.e., the induced subgraph containing any particular coordinate has good expansion (in fact is a clique). Lemma~\ref{lem:clique} guarantees that  and thus $\S$ admits a direct product theorem. Therefore, $\S$ demonstrates that direct product testing does not require the test graph to be an expander (like the Johnson/Ramanujan graph) but only to have certain local expansion properties.

\noindent\textbf{Sub-linear Size Domains.}
We remark here that we could consider subsets $\tilde{\S}$ of $\S$ of size smaller than $n$ which \emph{still} admit a direct product theorem. For example consider $\tilde{\S}$ as follows:
$$
\tilde{\S}=\{\{\nicefrac{ik}{2},\ldots ,\nicefrac{ik}{2}+k-1\}\mid i\in[2n/k]\},
$$
and the test graph $G_{\tilde{\S}}$ is given by the vertex set $\tilde \S$ and the edge set $\{(S,S')\mid S\cap S'\neq\emptyset\}$. It is easy to see that $\tilde{\S}$ admits a direct product theorem by applying Lemma~\ref{lem:clique}. Again, we emphasize that $|\tilde \S|=2n/k$ and yet admits a direct product theorem.

\noindent\textbf{Comparison with Dinur and Kaufman.}
One might wonder that if direct product testing results can be established on linear sized direct product domains using simple constructions such as the sliding window domain then, why did \cite{DK17} work so hard and use extremely heavy objects such as high dimensional expanders to obtain linear sized direct product domains. This is because for applications to gap and hardness amplification, it is desirable that a direct product domain also has distance amplification (defined below) and high dimensional expanders have distance amplification whereas the sliding window domain does not. 

\begin{definition}[Distance Amplification, \cite{DK17}]\label{def:dist}
A direct product domain $\D\subseteq \binom{[n]}{k}$ is said to have distance amplification  if for every $x,y\in\{0,1\}^n$ such that $\delta:=\Delta(x,y)<\nicefrac{1}{k}$, we have that $\Delta(\DP_V(x),\DP_V(y))=\Omega(k\delta)$. 
\end{definition}

Thus, the construction of the sliding window domain provides a conceptual clarification as to why we need high dimensional expanders: we can obtain direct product testing from simple constructions like the sliding window domain and we can obtain distance amplification from known constructions of vertex expanders (see Appendix~\ref{sec:distamp}); but to obtain both simultaneously, \cite{DK17} needed high dimensional expanders. 

\subsection*{Acknowledgments}
We are truly grateful to Irit Dinur for her constant support throughout this project and for her many illuminating and helpful discussions and comments. We also thank the anonymous reviewers for their detailed and useful feedback.

\bibliographystyle{alpha}
\bibliography{references}

\appendix

\section{Missing Proofs}\label{sec:appendix} 

\begin{proof}[Proof of Claim~\ref{claim:expanderMixing}]
		We prove only the furthermore part. The first part follows by plugging $c=1$.
				 For a set $S\subseteq V$ we denote by $\vecOne{S}$ the characteristic vector of $S$.
	Let $p_{\mu}\in \mathbb{R}^n$ be the distribution vector that describes $\mu$. First observe that:		
	$$\Pr_{(u,v)} [v\in T\mid u\sim\mu]=\frac 1 d \cdot (p_\mu)^t \cdot A \cdot \vecOne{ T}, $$
	where the probability is taken over $u$ that is drawn according to $\mu$ and $v$ is a uniformly random neighbor of $u$. Note that $\vecOne{V}$ is an eigenvector of $A$ corresponding to the largest eigenvalue (in absolute value) of $d$.
		We decompose the vectors: $p_\mu, \vecOne{T}$ as follows:
	
	\[ p_\mu = \frac{1}{\abs{V}} \vecOne{V}+\vec{p} \text{ and } \vecOne{{T}}= \gamma \vecOne{V}+\vec{t}.\]
	
	 Note that  $\gamma=\frac{\abs{ T}}{\abs{V}}$ and  $\vec{p}, \vec {t}$ are both orthogonal to $\vecOne{V}$ and let $\beta=\frac{\abs{S}}{\abs{V}}$.
	In these notations:
	\begin{eqnarray*}
		\frac{1}{d}\cdot (p_\mu)^t\cdot A \cdot \vecOne{ T} & = & \frac 1 d\left(\frac{1}{\abs{V}} \vecOne{V}+\vec{p} \right)^t \cdot A \cdot  ( \gamma \vecOne{V}+\vec{t}) \\
		& = & \gamma  + \ip{\vec{p}A}{\vec{t}} \\
		& \le & \gamma  +\frac \lambda d \norm {\vec{p}} \cdot  \norm{\vec{t}},
	\end{eqnarray*}
		where in the last step we used the Cauchy-Schwarz inequality and the fact that $\norm {\vec{p}A}\le \lambda \norm{\vec p}$. Now since the value of each coordinate of $p_{\mu}$ is upper bounded by $\frac c {\abs{S}}$ we get: $\norm {\vec{p}}^2=\norm{p_\mu}^2 - \frac{1}{\abs{V}^2}\norm{\vecOne{V}}^2\le \frac{c}{\abs{S}}-\frac{1}{\abs{V}}$, and $\norm {\vec{t}}^2=(\gamma(1-\gamma))\abs{V}$. So:
	\begin{eqnarray*}
		\Pr_{(u,v)} [v\in T\mid u\sim\mu]& = &\frac 1 d\cdot (p_\mu)^t\cdot A \cdot \vecOne{ T}\\
		 & \le & \gamma+ \frac{\lambda}d \cdot \sqrt{\left(\frac{c}{\abs{S}}-\frac{1}{\abs{V}}\right)\gamma(1-\gamma)\abs{V}}\\
		& \le & 
		\gamma +\frac \lambda d \cdot \sqrt{\frac{c\gamma \abs{V}}{\abs{S}}} \\
		 & = & \frac{|T|}{|V|} + \frac{\lambda}{d}\cdot \sqrt{\frac{c|T|}{|S|}}
	\end{eqnarray*}
	\end{proof}
	
\section{Application of Theorem~\ref{thm:expanders}: $\Omega(n)$-slice of the Hypercube}\label{sec:apphalf}

In this 
section, we consider the test graph $J(n,k,k(0.5-2\varepsilon))$, where $\varepsilon$ is some small constant. The domain of the direct product encoding is $\nk$. The pair $(S,S')$ is connected by an edge if and only if: $\abs {S\cap S'}=k(1/2-2\varepsilon)$.
We show that:
\begin{lemma}\label{claim:JhonsonEigen}
	Let $\varepsilon\in (0,1/33)$. Let $n\in \mathbb{N}$, let $1/2 -\varepsilon\le c\le 1/2$ be a constant and let $k=c \cdot n$, then the graph $J(n,k,k\cdot(1/2-2\varepsilon))$ is $(6\varepsilon,1/2-2\varepsilon)$-coordinate expander for large enough  $n$.
\end{lemma}
\begin{proof}
	\begin{enumerate}
		\item The proof of the second largest eigenvalue in absolute value of $J(n,k,t)$ was recently confirmed in \cite{BCIM18} to be as follows:
\begin{theorem}[Theorem 3.10 in  \cite{BCIM18}]
The second largest (normalized) eigenvalue of $J(n,k,t)$ is:
$$
\frac{\binom{k-1}{t-1}\cdot\binom{n-k}{k-t}-\binom{k}{t}\cdot\binom{n-k-1}{k-t-1}}{\binom{k}{t}\cdot\binom{n-k}{k-t}}=\frac{t}{k}-\frac{k-t}{n-k},
$$
whenever $(k-t)(n-1)\ge k(n-k)$.
\end{theorem}

		 Note that  for the value of $n,k,t$ we are interested in we have, $$k(n-k)\le nk(1/2+\varepsilon)=n(k-t-\varepsilon k)\le (n-1)(k-t),$$ when $n\ge \frac{1}{2\varepsilon}+2$.
		Therefore we can apply the above theorem and obtain a bound on the second largest eigenvalue in absolute value: 
\begin{align*}
|\lambda|&=\left\lvert\frac{1}{2}-2\varepsilon - \frac{(1/2+2\varepsilon)}{\frac{1}{c}-1}\right\rvert\\
&\le \frac{1}{2}\cdot \left\lvert 1-\frac{1}{\frac{1}{c}-1}\right\rvert + 2\varepsilon\cdot \left\lvert 1+\frac{1}{\frac{1}{c}-1}\right\rvert \\
&\le \frac{1}{2}\cdot \left\lvert \frac{4\varepsilon}{1+2\varepsilon}\right\rvert + 4\varepsilon\le 6\varepsilon
\end{align*}
				
		\item Fix $i\in [n]$. Then the graph $G_i$ is isomorphic to $J(n-1,k-1,k\cdot(1/2-2\varepsilon)-1)$. Therefore by the first item $|\lambda(G_i)|<|\lambda|$ as $\frac{t-1}{k-1}<\frac{t}{k}$. Clearly, for every value of $i\in [n]$ and for each $S\in V_i$ the probability that $i\in S'$ equals $1/2-2\varepsilon$.
		
		\item Fix $S,T\subseteq S$. Let $|T|=\alpha$.  Suppose $k-\alpha\ge 2\cdot k(1/2-2\varepsilon)$, i.e., $\alpha\le 4\varepsilon k$.
		\begin{eqnarray*}
			\Pr [\abs{T\cap S'} \le \alpha/4] & \le  & \frac{\sum_{i=0}^{\alpha/4}\binom{\alpha}{i}\cdot \binom{k-\alpha}{k(1/2-2\varepsilon)-i}}{\binom{k}{k(1/2-2\varepsilon)}} \\
			&\le & \frac{\binom{\alpha}{\alpha/4}\sum_{i=0}^{\alpha/4} \binom{k-\alpha}{k(1/2-2\varepsilon)-i}}{\binom{k}{k(1/2-2\varepsilon)}} \\
			&\le &\frac{\alpha/4\cdot \binom{\alpha}{\alpha/4}\cdot  \binom{k-\alpha}{k(1/2-2\varepsilon)}}{\binom{k}{k(1/2-2\varepsilon)}} \\
			&\le & 1.01\alpha/4\cdot (2^{H(0.25)\alpha}\cdot e^{-k^2(1/2-2\varepsilon)/\alpha}) \\
						&\le & 1.01\alpha/4\cdot (2^{H(0.25)\alpha-1.4k}) =o(1),
		\end{eqnarray*}
		for large $n$. On the other hand, suppose $\alpha> 4\varepsilon k$. Then we have, 
	\begin{eqnarray*}
			\Pr [\abs{T\cap S'} \le \alpha/4] & \le  & \frac{\binom{\alpha}{\alpha/4}\sum_{i=0}^{\alpha/4} \binom{k-\alpha}{k(1/2-2\varepsilon)-i}}{\binom{k}{k(1/2-2\varepsilon)}} \\
			& \le  & \frac{\binom{\alpha}{\alpha/4}\cdot 2^{k-\alpha}}{\binom{k}{k(1/2-2\varepsilon)}} \\
			& \le  & 2^{-(1-H(1/4))\alpha+(1-H(1/2 -2\varepsilon))k+o(k)}\\
			 			& \le  & 2^{-(1-H(1/4))\alpha+((\log_2 e)\cdot 16\varepsilon^2)k+o(k)},
		\end{eqnarray*}
		where we used the inequality that $H(1/2 -2\varepsilon)\ge 1-(\log_2 e)\cdot 16\varepsilon^2$. Therefore it suffices to show that $(1-H(1/4))\frac{\alpha}{k} - (16\varepsilon^2\cdot (\log_2 e))$ is a positive constant  bounded away from 0. We estimate $\log_2 e\le 1.45$ and $H(1/4)\le 0.82$. Thus we have,
$$		(1-H(1/4))\frac{\alpha}{k} - (16\varepsilon^2\cdot (\log_2 e)) >0 \Leftrightarrow 
	\frac{\alpha}{k}>129\varepsilon^2,$$
	and this follows since $\alpha>4\varepsilon k>129\varepsilon^2 k$ whenever $\varepsilon<1/33$.\qedhere\end{enumerate}
\end{proof}

As a corollary we get a direct product theorem when the domain $\D$ equals $\nk$ for values of $k$ which are close to $n/2$ (by applying Theorem~\ref{thm:expanders} keeping in mind Remark~\ref{rem:1}). Recall that \cite{DS14} established this result for all $k\le n/2$.

\section{Simple Applications of Lemma~\ref{lem:clique}}\label{sec:simple}

In this subsection,  we consider two direct product domains, namely $\binom{[n]}{n/2}$ and $\binom{[n]}{2}$ and prove a direct product theorem for these domains when the test graph is a clique and a member of Johnson graph family respectively.

\subsection{$\nicefrac{n}{2}$ slice of the Hamming cube}\label{sec:johnhalf}
A natural two-query test on the $\nicefrac{n}{2}$ slice of the Hamming cube is as follows:

\begin{center}
\fbox{%
	\parbox{0.75\textwidth}{%
		\noindent\textbf{Test $\mathcal{T}$}:
		\par
		\noindent \textbf{Input}: A function $F:\binom{[n]}{n/2} \to\{0,1\}^{\nicefrac{n}{2}}$.
		\par
		\noindent \textbf{Procedure}: 
		 Pick uniformly and independently at random $S,S'\in\binom{[n]}{n/2}$. 
		\par		
		\noindent \textbf{Output}: Accept if and only if $F(S)|_{S\cap S'}=F(S')|_{S\cap S'}$.	
		}}
		\end{center}
		
We now interpret the above test in the language established in Section~\ref{sec:setting}. In the above test, the domain $V$ of the direct product is $\binom{[n]}{n/2}$ and the test graph $G$ is a clique with self loops. Therefore, for every $i\in[n]$, the induced subgraph of $V_i$ in $G$ is a clique (with self loops). And, for every $S\in V$ and every $i\in S$, the probability that a uniformly random neighbor $S'$ of $S$ in $G$ contains $i$ is $\nicefrac{1}{2}$. Thus, from Lemma~\ref{lem:clique} we have that for any $F\in\mathcal F_V$, if $F$ passes $\Test(G)$ with probability $1-\varepsilon$ then $F$ is $(1-4\varepsilon)$-close to $\d(F)$.

\subsection{$J(n,2,1)$ of the Johnson Graph Family}

For the domain $\binom{[n]}{2}$, we note that if we pick two elements from $\binom{[n]}{2}$ uniformly and independently at random then they have empty intersection with probability almost 1. Therefore, the same test as for the $\nicefrac{n}{2}$ slice of the Hamming cube does not work here. Nonetheless, there is still a natural two-query test for the domain $\binom{[n]}{2}$ described as follows:

\begin{center}
\fbox{%
	\parbox{0.6\textwidth}{%
		\noindent\textbf{Test $\mathcal{T}$}:
		\par
		\noindent \textbf{Input}: A function $F:\binom{[n]}{2} \to\{0,1\}^{2}$.
		\par
		\noindent \textbf{Procedure}: 
		 Pick uniformly at random $S\in\binom{[n]}{2}$. Then pick uniformly at random $S'\in \binom{[n]}{2}$ such that $|S\cap S'|=1$.
		\par		
		\noindent \textbf{Output}: Accept if and only if $F(S)|_{S\cap S'}=F(S')|_{S\cap S'}$.	
		}}\end{center}

		 		We now interpret the above test in the language established in Section~\ref{sec:setting}. In the above test, the domain $V$ of the direct product is $\binom{[n]}{2}$ and the test graph $G$ is $J(n,2,1)$. Note that for every $i\in[n]$, the induced subgraph of $V_i$ in $G$ is a clique (without self loops) because any two distinct subsets in $V_i$ have $i$ in their intersection and thus have intersection size equal to 1. Also for every $S\in V$ and every $i\in S$, the probability that a uniformly random neighbor $S'$ of $S$ in $G$ contains $i$ is $\nicefrac{1}{2}$. Thus, from Lemma~\ref{lem:clique} we have that for any $F\in\mathcal F_V$, if $F$ passes $\Test(G)$ with probability $1-\varepsilon$ then $F$ is $(1-4\varepsilon)$-close to $\d(F)$.

\section{Linear Sized Domains having Distance Amplification}\label{sec:distamp}

In this section, we show how to construct a collection of sets which have distance amplification. To do so we rely on the existence of vertex expanders. 

\begin{definition}[Vertex Expansion]
Let $G(V,E)$ be a $d$-regular graph. For every subset $S\subseteq V$ let $\partial(S)=\{u\in V\setminus S\mid \exists v\in S\text{ such that }(u,v)\in E\}$. The vertex isoperimetric constant $h(G)$ is defined as follows:
$$
h(G)=\min_{0\le |S|\le |V|/d} \frac{|\partial(S)|}{|S|\cdot d}.
$$
We say that $G$ is a vertex expander if $h(G)$ is a constant bounded away from 0.
\end{definition}

\begin{theorem}[Folklore]\label{thm:Vexp}
For all $d>2$, a random $d$-regular graph is a vertex expander with high probability.
\end{theorem}

Given a $d$-regular graph $G(V,E)$ (where $n:=|V|$) which is a vertex expander with vertex isoperimetric constant $\gamma>0$, we show how to construct $\S_G\subseteq \binom{[n]}{d}$ of cardinality $n$ such that $\S_G$ has distance amplification. We identify the vertices in $V$ with $[n]$ and construct $\S_G$ as follows:
$$
\S_G=\{\partial(\{v\})\mid v\in V\}.
$$

\begin{claim}
$\S_G$ has distance amplification.
\end{claim}
\begin{proof}
Fix distinct $x,y\in\{0,1\}^n$. Let $\delta:=\Delta(x,y)\le \nicefrac{1}{d}$. Let $R\subseteq [n]$ be the set of coordinates on which $x$ and $y$ differ. Clearly, $|R|\le n/d$. The number of subsets in $\S_G$ that contain an element in $R$ is at least $\gamma d |R|$. Therefore we have $\Delta(\DP_{\S_G}(x),\DP_{\S_G}(y))\ge \gamma \delta d$.
\end{proof}

\end{document}